\documentclass[final]{article}

\usepackage{adjustbox}
\usepackage{subcaption}
\usepackage{algorithm}
\usepackage[noend]{algpseudocode}
\usepackage{amsmath}
\usepackage{amsthm}
\usepackage{tikz}
\usepackage{url}
\usepackage{hyperref}

\theoremstyle{plain}
\newtheorem{thm}{Theorem}
\numberwithin{equation}{section}
\numberwithin{figure}{section}
\numberwithin{table}{section}
\numberwithin{algorithm}{section}

\DeclareMathOperator{\divi}{div}
\DeclareMathOperator{\md}{{\mathrm{mod}}}

\title{Minimizing communication in the multidimensional FFT}
\author{Thomas Koopman\thanks{Software Science, Radboud University,
P.O. Box 9010, 6500 GL Nijmegen, The Netherlands (\texttt{thomas.koopman@ru.nl}).}
    and Rob H. Bisseling\thanks{Mathematical Institute, Utrecht University,
P.O. Box 80010, 3508 TA Utrecht, The Netherlands (\texttt{R.H.Bisseling@uu.nl}).}}

\tikzset{
  redSquare/.pic={
    \draw[fill=red] (0,0,0) rectangle (1,1,0);
  }
}

\tikzset{
  blueSquare/.pic={
    \draw[fill=blue] (0,0,0) rectangle (1,1,0);
  }
}

\tikzset{
  yellowSquare/.pic={
    \draw[fill=yellow] (0,0,0) rectangle (1,1,0);
  }
}

\tikzset{
  greenSquare/.pic={
    \draw[fill=green] (0,0,0) rectangle (1,1,0);
  }
}

\tikzset{
  redCube/.pic={
    \draw[fill=red!20] (0,1,0) -- (0,1,1) -- (1,1,1) -- (1,1,0);
    \draw[fill=red!50] (1,0,0) -- (1,0,1) -- (1,1,1) -- (1,1,0);
    \draw[fill=red] (0,0,0) rectangle (1,1,0);
  }
}

\tikzset{
  blueCube/.pic={
    \draw[fill=blue!20] (0,1,0) -- (0,1,1) -- (1,1,1) -- (1,1,0);
    \draw[fill=blue!50] (1,0,0) -- (1,0,1) -- (1,1,1) -- (1,1,0);
    \draw[fill=blue] (0,0,0) rectangle (1,1,0);
  }
}

\tikzset{
  yellowCube/.pic={
    \draw[fill=yellow!20] (0,1,0) -- (0,1,1) -- (1,1,1) -- (1,1,0);
    \draw[fill=yellow!50] (1,0,0) -- (1,0,1) -- (1,1,1) -- (1,1,0);
    \draw[fill=yellow] (0,0,0) rectangle (1,1,0);
  }
}

\tikzset{
  greenCube/.pic={
    \draw[fill=green!20] (0,1,0) -- (0,1,1) -- (1,1,1) -- (1,1,0);
    \draw[fill=green!50] (1,0,0) -- (1,0,1) -- (1,1,1) -- (1,1,0);
    \draw[fill=green] (0,0,0) rectangle (1,1,0);
  }
}

\tikzset{
  purpleCube/.pic={
    \draw[fill=purple!20] (0,1,0) -- (0,1,1) -- (1,1,1) -- (1,1,0);
    \draw[fill=purple!50] (1,0,0) -- (1,0,1) -- (1,1,1) -- (1,1,0);
    \draw[fill=purple] (0,0,0) rectangle (1,1,0);
  }
}

\tikzset{
  orangeCube/.pic={
    \draw[fill=orange!20] (0,1,0) -- (0,1,1) -- (1,1,1) -- (1,1,0);
    \draw[fill=orange!50] (1,0,0) -- (1,0,1) -- (1,1,1) -- (1,1,0);
    \draw[fill=orange] (0,0,0) rectangle (1,1,0);
  }
}

\tikzset{
  pinkCube/.pic={
    \draw[fill=pink!20] (0,1,0) -- (0,1,1) -- (1,1,1) -- (1,1,0);
    \draw[fill=pink!50] (1,0,0) -- (1,0,1) -- (1,1,1) -- (1,1,0);
    \draw[fill=pink] (0,0,0) rectangle (1,1,0);
  }
}

\tikzset{
  blackCube/.pic={
    \draw[fill=black!20] (0,1,0) -- (0,1,1) -- (1,1,1) -- (1,1,0);
    \draw[fill=black!50] (1,0,0) -- (1,0,1) -- (1,1,1) -- (1,1,0);
    \draw[fill=black] (0,0,0) rectangle (1,1,0);
  }
}

\newcommand{\cyclicOneD}[1]{
\begin{tikzpicture}[z = {(0.4, 0.4)}]
\path
\foreach \a in {0, ..., #1} {
  	(4 * \a, 0)pic{greenSquare}
  	(4 * \a + 1, 0)pic{yellowSquare}
    (4 * \a + 2, 0)pic{blueSquare}
    (4 * \a + 3, 0)pic{redSquare}
}
;
\end{tikzpicture}
}

\newcommand{\cyclicTwoD}[2]{
\begin{tikzpicture}[z = {(0.4, 0.4)}]
\path
\foreach \a in {0, ..., #1} {
	\foreach \b in {0, ..., #2} {
    	(2 * \a, 2 * \b)pic{greenSquare}
    	(2 * \a + 1, 2 * \b)pic{yellowSquare}
		(2 * \a, 2 * \b + 1)pic{blueSquare}
		(2 * \a + 1, 2 * \b + 1)pic{redSquare}
	}
}
;
\end{tikzpicture}
}

\newcommand{\cyclicThreeD}[3]{
\begin{tikzpicture}[z = {(0.34, 0.4)}]
\path
\foreach \c in {#3, ..., 0} {
	\foreach \a in {0, ..., #1} {
		\foreach \b in {0, ..., #2} {
    		(2 * \a, 2 * \b, 2 * \c + 1)pic{orangeCube}
    		(2 * \a + 1, 2 * \b, 2 * \c + 1)pic{yellowCube}
    		(2 * \a, 2 * \b + 1, 2 * \c + 1)pic{blueCube}
			(2 * \a + 1, 2 * \b + 1, 2 * \c + 1)pic{redCube}
    	}
    }
    \foreach \a in {0, ..., #1} {
    	\foreach \b in {0, ..., #2} {
        	(2 * \a, 2 * \b, 2 * \c)pic{greenCube}
    		(2 * \a + 1, 2 * \b, 2 * \c)pic{purpleCube}
    		(2 * \a, 2 * \b + 1, 2 * \c)pic{pinkCube}
			(2 * \a + 1, 2 * \b + 1, 2 * \c)pic{blackCube}
		}
	}
};
\end{tikzpicture}
}

\newcommand{\drawOffsetCube}[7] {
\foreach \c in {#3, ..., 0} {
	\foreach \a in {0, ..., #1} {
		\foreach \b in {0, ..., #2} {
    		(\a + #4*#1 + #4, \b + #5*#2 + #5, \c + #6*#3 + #6)pic{#7}
    	}
    }
}
}

\newcommand{\slabThreeDz}[3]{
\begin{tikzpicture}[z = {(0.34, 0.4)}]
\path
\drawOffsetCube{#1}{#2}{#3}{0}{0}{7}{blackCube}
\drawOffsetCube{#1}{#2}{#3}{0}{0}{6}{yellowCube}
\drawOffsetCube{#1}{#2}{#3}{0}{0}{5}{blueCube}
\drawOffsetCube{#1}{#2}{#3}{0}{0}{4}{redCube}
\drawOffsetCube{#1}{#2}{#3}{0}{0}{3}{greenCube}
\drawOffsetCube{#1}{#2}{#3}{0}{0}{2}{purpleCube}
\drawOffsetCube{#1}{#2}{#3}{0}{0}{1}{pinkCube}
\drawOffsetCube{#1}{#2}{#3}{0}{0}{0}{orangeCube};
\end{tikzpicture}
}

\newcommand{\slabThreeDx}[3]{
\begin{tikzpicture}[z = {(0.34, 0.4)}]
\path
\drawOffsetCube{#1}{#2}{#3}{0}{0}{0}{blackCube}
\drawOffsetCube{#1}{#2}{#3}{1}{0}{0}{yellowCube}
\drawOffsetCube{#1}{#2}{#3}{2}{0}{0}{blueCube}
\drawOffsetCube{#1}{#2}{#3}{3}{0}{0}{redCube}
\drawOffsetCube{#1}{#2}{#3}{4}{0}{0}{greenCube}
\drawOffsetCube{#1}{#2}{#3}{5}{0}{0}{purpleCube}
\drawOffsetCube{#1}{#2}{#3}{6}{0}{0}{pinkCube}
\drawOffsetCube{#1}{#2}{#3}{7}{0}{0}{orangeCube};
\end{tikzpicture}
}

\newcommand{\slabThreeDy}[3]{
\begin{tikzpicture}[z = {(0.34, 0.4)}]
\path
\drawOffsetCube{#1}{#2}{#3}{0}{0}{0}{blackCube}
\drawOffsetCube{#1}{#2}{#3}{0}{1}{0}{yellowCube}
\drawOffsetCube{#1}{#2}{#3}{0}{2}{0}{blueCube}
\drawOffsetCube{#1}{#2}{#3}{0}{3}{0}{redCube}
\drawOffsetCube{#1}{#2}{#3}{0}{4}{0}{greenCube}
\drawOffsetCube{#1}{#2}{#3}{0}{5}{0}{purpleCube}
\drawOffsetCube{#1}{#2}{#3}{0}{6}{0}{pinkCube}
\drawOffsetCube{#1}{#2}{#3}{0}{7}{0}{orangeCube};
\end{tikzpicture}
}

\newcommand{\pencilThreeDyz}[3]{
\begin{tikzpicture}[z = {(0.34, 0.4)}]
\path
\drawOffsetCube{#1}{#2}{#3}{0}{0}{3}{blackCube}
\drawOffsetCube{#1}{#2}{#3}{0}{0}{2}{yellowCube}
\drawOffsetCube{#1}{#2}{#3}{0}{0}{1}{blueCube}
\drawOffsetCube{#1}{#2}{#3}{0}{0}{0}{redCube}
\drawOffsetCube{#1}{#2}{#3}{0}{1}{3}{greenCube}
\drawOffsetCube{#1}{#2}{#3}{0}{1}{2}{purpleCube}
\drawOffsetCube{#1}{#2}{#3}{0}{1}{1}{pinkCube}
\drawOffsetCube{#1}{#2}{#3}{0}{1}{0}{orangeCube};
\end{tikzpicture}
}

\newcommand{\pencilThreeDxy}[3]{
\begin{tikzpicture}[z = {(0.34, 0.4)}]
\path
\drawOffsetCube{#1}{#2}{#3}{0}{0}{0}{blackCube}
\drawOffsetCube{#1}{#2}{#3}{0}{1}{0}{yellowCube}
\drawOffsetCube{#1}{#2}{#3}{0}{2}{0}{blueCube}
\drawOffsetCube{#1}{#2}{#3}{0}{3}{0}{redCube}
\drawOffsetCube{#1}{#2}{#3}{1}{0}{0}{greenCube}
\drawOffsetCube{#1}{#2}{#3}{1}{1}{0}{purpleCube}
\drawOffsetCube{#1}{#2}{#3}{1}{2}{0}{pinkCube}
\drawOffsetCube{#1}{#2}{#3}{1}{3}{0}{orangeCube};
\end{tikzpicture}
}

\newcommand{\pencilThreeDxz}[3]{
\begin{tikzpicture}[z = {(0.34, 0.4)}]
\path
\drawOffsetCube{#1}{#2}{#3}{0}{0}{3}{blackCube}
\drawOffsetCube{#1}{#2}{#3}{0}{0}{2}{yellowCube}
\drawOffsetCube{#1}{#2}{#3}{0}{0}{1}{blueCube}
\drawOffsetCube{#1}{#2}{#3}{0}{0}{0}{redCube}
\drawOffsetCube{#1}{#2}{#3}{1}{0}{3}{greenCube}
\drawOffsetCube{#1}{#2}{#3}{1}{0}{2}{purpleCube}
\drawOffsetCube{#1}{#2}{#3}{1}{0}{1}{pinkCube}
\drawOffsetCube{#1}{#2}{#3}{1}{0}{0}{orangeCube};
\end{tikzpicture}
}

\newcommand{\Axes}[3]{
\begin{tikzpicture}[z = {(0.34, 0.4)}]
\draw [->,thick] (0,0) -- (1,0) ;
\draw [->,thick] (0,0) -- (0,1) ; 
\draw [->,thick] (0,0) -- (0.34,0.4); 
\node [right] at (1,0) {$x$} ;
\node [right] at (0.34,0.4) {$y$} ;
\node [above] at (0,1) {$z$} ;
\end{tikzpicture}
}

\begin{document}
\maketitle

\begin{abstract}
We present a parallel algorithm for the fast Fourier transform (FFT) in higher dimensions. This algorithm generalizes the cyclic-to-cyclic one-dimensional parallel algorithm to a cyclic-to-cyclic  multidimensional parallel algorithm while retaining the property of needing only a single all-to-all communication step. This is under the constraint that we use at most $\sqrt{N}$ processors for an FFT on an array
with a total of $N$ elements, irrespective of the dimension $d$ or the shape of the array.
The only assumption we make is that $N$ is sufficiently composite.
Our algorithm starts and ends in the same data distribution.

We present our multidimensional implementation FFTU which utilizes the sequential FFTW program for its local FFTs,
and which can handle any dimension $d$. We obtain experimental results for $d\leq 5$ using MPI
on up to 4096 cores of the supercomputer Snellius, comparing  FFTU with the parallel FFTW program and with PFFT and heFFTe.
These results show that
FFTU is competitive with the state of the art and that it allows one to use a larger number of processors,
while keeping communication limited to a single all-to-all operation.
For arrays of size $1024^3$ and $64^5$, FFTU achieves a speedup of a factor 149 and 176, respectively, on 4096 processors.
\end{abstract}

\section{Introduction}

The one-dimensional discrete Fourier transform (1D DFT) of length $n$ is a function $F_n$ that takes an array $x$ of $n$ complex numbers, and returns an array $y = F_n(x)$ of $n$ complex numbers. We write $\omega_n$ for the $n$th root of unity $e^{-2\pi i / n}$ and $[n]$ for $\{0, \ldots, n - 1\}$. The DFT is given by
\begin{equation}
\label{eq:dft}
y_k = \sum_{j \in [n]} x_j \omega_n^{jk},\quad \mathrm{for}~ k \in [n].
\end{equation}
The fast Fourier transform (FFT)~\cite{cooley65} is a fast implementation of the DFT that can perform the computation in $\mathcal{O} (n \log n)$ time
instead of the $\mathcal{O} (n^2)$ time of a straightforward implementation of (\ref{eq:dft}); see~\cite{vanloan92} for
a detailed exposition.

We also have a multidimensional variant of the DFT, which takes $n_1 \times \cdots \times n_d$ multidimensional arrays of complex numbers as both input and output, and is given by
\begin{align}
\label{eq:dftd}
Y[k_1,  \ldots , k_d] = \sum_{j_1 \in [n_1]}  \cdots  \sum_{j_d \in [n_d]} X[j_1,  \ldots , j_d] \omega_{n_1}^{j_1k_1} \cdots  \omega_{n_d}^{j_dk_d},\\
\hfill   \mathrm{for}~
(k_1,\ldots,k_d) \in [n_1]\times \cdots \times [n_d]. \nonumber
 \end{align}

We find it convenient in our notation to number the dimensions of the FFT starting at 1, but to have all other numberings start at 0.
Without loss of generality, we assume that $n_1 \geq n_2 \geq \cdots \geq n_d$.
We write $N= n_1\cdots n_d$ for the total number of array elements.

The multidimensional FFT is the main computational kernel in many applications. An example is
the solution of the time-dependent Schr\"{o}dinger equation by wave packet propagation~\cite{kosloff88,leforestier91},
where the FFT is used in a spectral method to compute the kinetic-energy operation efficiently,
and where the dimension increases rapidly with the number of atoms (three per additional atom);
see~\cite{borowski04} for a 6D calculation.
Other examples are classical molecular dynamics (MD), where a 3D FFT is used to compute
long-range interactions efficiently based on Ewald sums, for instance, in the LAMMPS package~\cite{plimpton97}, and
3D computed tomography~\cite{pryor17}.
These are just a few examples; there are numerous other applications
ranging from image processing to fluid flow simulation to
weather and climate prediction.

\subsection{Distributions}

Informally, a data distribution, or \emph{distribution},  is a specification on how to store a data structure divided over a set of processors. Suppose we have a global data structure $X$ indexed by some set $J$, e.g., $J = [n]$ if $X$ is a 1D array of size $n$, or $J = [n_1] \times \cdots \times [n_d]$ for $X$ a $d$-dimensional array of size $n_l$ in dimension $l$. Let $\{P(s) \ | \ s \in S\}$ be a set of processors. We denote the local data structure on $P(s)$ by $X^{(s)}$ and assume it is indexed by a set $J_s$. A distribution is then more formally a bijection $\phi: \sqcup_{s \in S} J_s \to J$ from the disjoint union of local index sets to the global index set,
where $\phi$ is defined by $X^{(s)}[k]$ corresponding to $X[\phi(s, k)]$.

The \emph{cyclic distribution} of a 1D array of length $n$ over a set of processors indexed by $[p]$ is given by $\phi(s, k) = s + kp$.
We generalize this to a $d$-dimensional cyclic distribution over a set of processors indexed by $[p_1] \times \cdots \times [p_d]$ for $d$-dimensional arrays. The distribution is then given by $\phi((s_1, \ldots, s_d), (k_1, \ldots, k_d)) = (s_1 + k_1p_1, \ldots, s_d + k_dp_d)$. When the dimension is clear, we will abbreviate this to $\phi(s, k) = s + kp$.
The cyclic distribution has been shown to be advantageous for the 1D parallel FFT~\cite{inda01},
because it can be used as both the input and the output distribution for the FFT while requiring
only one data redistribution during the computation; see also~\cite{bisseling20} for an extensive discussion.
The cyclic distribution can be used for up to $p = \sqrt{n}$ processors.
It has been shown to achieve state-of-the-art performance in a 1D implementation that uses
FFTW for sequential FFTs and MulticoreBSP for C for communication~\cite{yzelman14}.
The 1D, 2D, and 3D cyclic distributions for 1D, 2D, and 3D arrays, respectively, are illustrated in Fig.~\ref{fig:cyclic_distributions}.

\begin{figure}[ht]
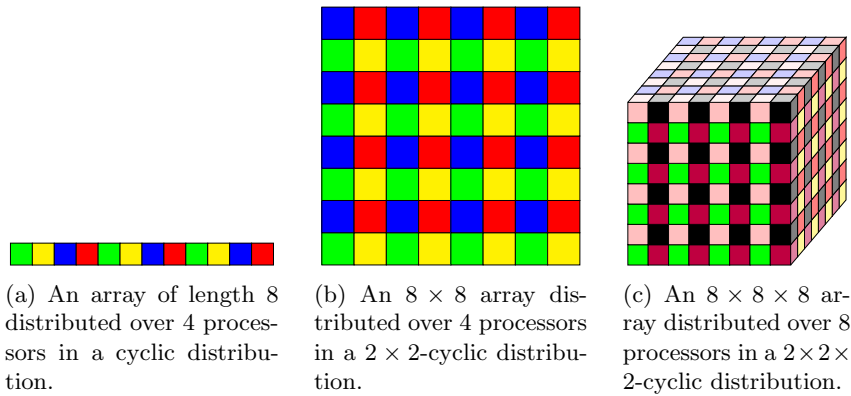

    \centering
     \begin{subfigure}[b]{0.3\textwidth}
        \centering
	    \adjustbox{max width=\textwidth}{
	        \cyclicOneD{2}
        }
	    \caption{An array of length $8$ distributed over $4$ processors in a cyclic distribution.}
    	\label{fig: 1d cyclic distribution}
    \end{subfigure}
    \quad
    \begin{subfigure}[b]{0.3\textwidth}
        \centering
	    \adjustbox{max width=\textwidth}{
	        \cyclicTwoD{3}{3}
        }
	    \caption{An $8 \times 8$ array distributed over $4$ processors in a $2 \times 2$-cyclic distribution.}
	    \label{fig:2d_cyclic_distribution}
    \end{subfigure}
    \quad
     \begin{subfigure}[b]{0.25\textwidth}
        \centering
	    \adjustbox{max width=\textwidth}{
	        \cyclicThreeD{3}{3}{3}
        }
	    \caption{An $8 \times 8 \times 8$ array distributed over $8$ processors in a $2 \times 2 \times 2$-cyclic distribution.}
	    \label{fig:3d_cyclic_distribution}
    \end{subfigure}

    \caption{Cyclic distribution in several dimensions, indicated by colors.}
    \label{fig:cyclic_distributions}
\end{figure}

The \emph{slab distribution} (or \emph{1D distribution}) distributes a multidimensional array by blocks along one dimension. If we have $p$ processors
and distribute along the first dimension, the distribution is given by $\phi(s, (k_1, \ldots, k_d)) = (k_1 + s \cdot \frac{n}{p}, k_2, \ldots, k_d)$. Distribution along another dimension can be done analogously.
This is illustrated in Fig.~\ref{fig:3d_slab_distribution}.

\begin{figure}[ht]
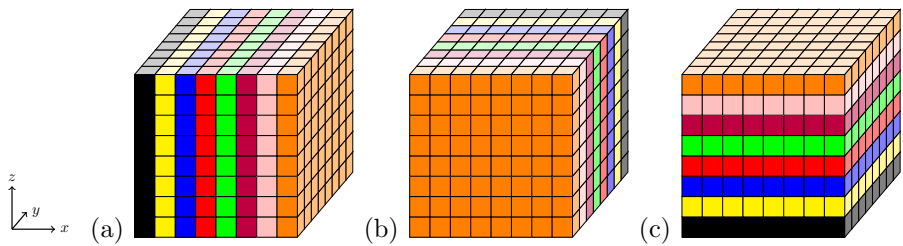

	\adjustbox{max width=0.1\textwidth}{
        \Axes{0}{0}{0}
	}
        (a)
	\adjustbox{max width=0.25\textwidth}{
        \slabThreeDx{0}{7}{7}
	}
	(b)
	\adjustbox{max width=0.25\textwidth}{
        \slabThreeDz{7}{7}{0}
        }
        (c)
        \adjustbox{max width=0.25\textwidth}{
        \slabThreeDy{7}{0}{7}
	}
	\caption{An $8 \times 8 \times 8$ array distributed over $8$ processors in a slab distribution along different dimensions.
	(a) A slab distribution along the $x$-direction;
	(b) along the $y$-direction;
	(c) along the $z$-direction.}
	\label{fig:3d_slab_distribution}
\end{figure}

The \emph{pencil distribution} (or \emph{2D distribution}) is similar to the slab distribution, but it distributes along two dimensions. If the processor indexing set is $[p_1] \times [p_2]$ with $p=p_1p_2$, the pencil distribution is given by $\phi((s_1, s_2), (k_1, \ldots, k_d)) = (k_1 + s_1 \cdot \frac{n_1}{p_1}, k_2 + s_2 \cdot \frac{n_2}{p_2}, k_3, \ldots, k_d)$.
This is illustrated in Fig.~\ref{fig:3d_pencil_distribution}.
The pencil distribution can easily be generalized to a higher-dimensional distribution.

\begin{figure}[ht]
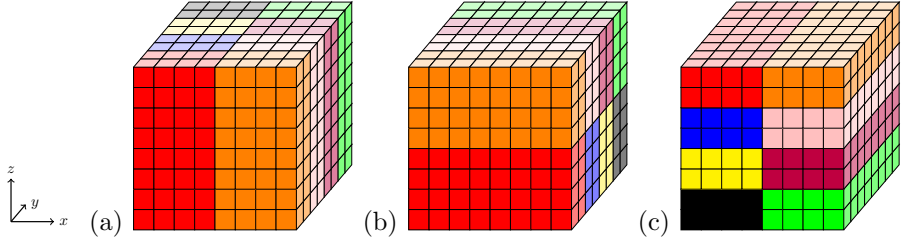

        \adjustbox{max width=0.1\textwidth}{
        \Axes{0}{0}{0}
	}
         (a)
	\adjustbox{max width=0.25\textwidth}{
        \pencilThreeDxz{3}{7}{1}
	}
	(b)
	\adjustbox{max width=0.25\textwidth}{
        \pencilThreeDyz{7}{3}{1}
	}
	(c)
	\adjustbox{max width=0.25\textwidth}{
        \pencilThreeDxy{3}{1}{7}
	}
	\caption{An $8 \times 8 \times 8$ array distributed over $2 \times 4$ processors in a pencil distribution along different dimensions.
	(a) Distribution along $x$ and $y$;
	(b) along $z$ and $y$;
	(c) along $x$ and $z$.
	}
	\label{fig:3d_pencil_distribution}
\end{figure}

\subsection{Related work}\label{sec:literature}

Many state-of-the-art parallel FFT libraries, including the packages FFTW~\cite{frigo05}, PFFT~\cite{pippig13}, and HeFFTe~\cite{ayala20} that we will use in our comparisons,
exploit that the sum describing the FFT factorizes as follows:

\begin{multline}
\sum_{j_1 \in [n_1]  }  \cdots  \sum_{j_d \in [n_d] }  X[j_1,  \ldots , j_d] \omega_{n_1}^{j_1k_1} \cdots  \omega_{n_d}^{j_dk_d} = \\
\sum_{j_1  \in [n_1]} \left( \sum_{j_2 \in [n_2]}  \cdots  \sum_{j_d \in [n_d]} X[j_1,  \ldots , j_d] \omega_{n_2}^{j_2k_2}\cdots   \omega_{n_d}^{j_dk_d} \right) \cdot \omega_{n_1}^{j_1k_1}.
\end{multline}

We see that the $d$-dimensional Fourier transform can be calculated in two steps: first, we apply a $(d - 1)$-dimensional FFT to $X[j_1, \ast, \ldots, \ast]$ for all $j_1$, calling the result $Y$ (usually this is done in place, so with the same data structure for $X$ and $Y$). Then we transform $Y[\ast, k_2, \ldots, k_d]$ by a 1D FFT for all tuples $(k_2, \ldots, k_d)$. This can be applied recursively to calculate a $d$-dimensional FFT by performing 1D transforms along each dimension.
This can be done in parallel using sequential 1D FFTs. See \cite{foster97} for alternative parallel methods such as the binary exchange method and a comparison between these methods.

The approach that parallel FFTW~\cite{frigo05} takes to parallelizing the multidimensional computation is as follows: it starts with a slab distribution and performs a sequential FFT in all directions that are local. In Fig.~\ref{fig:3d_slab_distribution}(a),
this is along the $y$-axis and the $z$-axis. Then it performs a communication step to move to a distribution such that the remaining direction becomes local, for instance, the slab distribution of Fig.~\ref{fig:3d_slab_distribution}(b).
Performing the final step in the slab distribution is not always possible. For example, for a parallel FFT
on an array of size $8 \times 4 \times 2$
that starts  with a slab distribution for 8 processors along the first dimension,
the final step of FFTW uses a $4 \times 2$ pencil distribution.
If needed, the final distribution will have an even higher dimension $r>2$.
Therefore, starting in a slab distribution along the first direction,
with the largest size,
FFTW can use at most $p_{\max} = \min(n_1, n_2\cdots n_d) = \min(n_1, \frac{N}{n_1}) $ processors.
This implies that in the most favorable case, when $n_1=\sqrt{N}$, FFTW can use $ p_{\max} = \sqrt{N}$ processors,
but in all other cases it must use fewer processors.

If we want to use more processors, we can choose to start in a pencil distribution instead of a slab distribution. This
approach was been proposed by Ding, Ferraro, and Gennery~\cite{ding95}.
It has been implemented in 3D packages by Plimpton (see \cite{plimpton97} and also~\cite{plimpton18}),
2DECOMP\&FFT by Li and Laizet~\cite{li10}, and P3DFFT by Pekurovsky~\cite{pekurovsky12}.
PFFT \cite{pippig13} is also based on using a pencil distribution, but generalized to
arbitrary $d \geq 3$.

Having a 2D processor distribution gives us $d - 2$ directions that are locally available. This means we have to switch pencil distributions
$\lceil \frac{d}{d - 2} \rceil - 1 = \lceil \frac{2}{d - 2} \rceil$ times to transform along all $d$ directions.
For $d=3$, this would mean switching twice, as illustrated by Fig.~\ref{fig:3d_pencil_distribution}, and for $d \geq 4$, this would be only once.
For $d=3$, we can use $p_{\max} = \min(n_1n_2,n_2n_3,n_1n_3) = n_2n_3=\frac{N}{n_1}$ processors,
by our assumption on the decreasing sizes of the different dimensions.
In the most favorable case, when all sizes are equal, $p_{\max} = N^{2/3}$; e.g.,
in Fig.~\ref{fig:3d_pencil_distribution}, we could have used up to 64 processors.

PFFT  generalizes earlier work by enabling the use of an $r$-dimensional data distribution with $1 \leq r<d$, not just $r=1,2$,
at the expense of $\lceil \frac{d}{d-r} \rceil - 1 =  \lceil \frac{r}{d - r} \rceil$ data redistributions.
For $d \leq 3$, the slab and pencil distributions suffice; here, we discuss the remaining case $d\geq 4$.
In general, distributing along $r$ dimensions makes the other $d-r$ dimensions locally available.
Assume that the dimensions chosen for distribution have sizes $m_1,\ldots, m_r$ and the others have sizes $m_{r+1},\ldots, m_d$;
thus the $m_l$ are a permutation of the $n_l$.
As a consequence, there are $m_{1} \cdots m_{r}$ FFTs, which gives an upper bound on the  number of processors we can use.
We are left with $m_{r + 1} \cdots m_d$ FFTs  to perform,
which reduces the upper bound to $\min(m_1 \cdots m_{r}, m_{r + 1} \cdots m_d)$ processors.
If we limit ourselves to requiring a single redistribution, this number equals $p_{\max}$,
and in that case we must choose $r \leq d-r$, i.e., $r\leq d/2$.
We can maximize $p_{\max}$ by choosing the
dimensions to be distributed initially such that $m_1 \cdots m_{r}$ is as close  as possible to $m_{r + 1} \cdots m_d$.
For even $d$, and in case all $n_l$ are the same, this means that $p_{\max} = N^{1/2}$ processors.
For all other cases, $p_{\max} \leq N^{1/2}$ is an upper bound.
For odd $d$, and in case all $n_l$ are the same, we can maximize $p_{\max} $ by taking $r=(d-1)/2$,
giving $p_{\max} = N^{(d-1)/(2d)} <N^{1/2}$; e.g., for $d=5$, $p_{\max} = N^{2/5}$.
Thus, if all sizes are equal, having an odd number of dimensions is less favorable.

Both FFTW and PFFT finish their computation in a different distribution from the input distribution.
If this is undesirable, an extra communication step is necessary.
This may be avoided, for instance, when the FFT is followed by a local elementwise computation (such as happens 
in a convolution involving an elementwise multiplication),
and an inverse FFT.
If this inverse FFT can start in the output distribution of the forward FFT and can finish in the original input distribution,
then the extra communication step can be avoided.

Jung et al.~\cite{jung16} present a 3D FFT with a \emph{volumetric} (3D) data distribution for use in an MD application.
In one of their schemes, 1D all-to-all, they perform parallel 1D FFTs within every direction, with five all-to-all communication steps
within one dimension. In an alternative scheme, 2D all-to-all, they require three all-to-all communication steps,
each within one or two dimensions. Because of the requirements of the
MD application of which the FFT is a part, the 3D distribution is by blocks.

Dalcin, Mortensen, and Keyes~\cite{dalcin19} provide new methods for accelerating the all-to-all communication step
of multidimensional FFTs, where they make use of advanced features of MPI-2.

Popovici et al.~\cite{popovici20} generalize the 1D FFT based on the cyclic distribution to a multidimensional FFT
based on a cyclic distribution in all dimensions with the advantage that the input and output distributions are the same.
They argue that to scale up to the large numbers of processors that are available in today's supercomputers,
it it necessary to use parallelism both across 1D FFTs and within 1D FFTs.
In each dimension $l$, their algorithm can then use  $p_l \leq  \sqrt{n_l}$ processors because of the cyclic distribution.
This implies that if all $n_l$ are squares, the maximum number of processors used equals  $p_{\max} = N^{1/2}$ .
Their algorithm performs a communication step that moves all data, once for every dimension, so the algorithm has
$d$ communication steps in total.
In their implementation, they use FFTW for the local computations, MPI for distributed-memory
parallelism, and OpenMP for shared-memory parallelism.
They present results for 3D arrays of size up to $1024^3 $.

The recent package heFFTe~\cite{ayala20} aims at providing highly efficient parallel 3D FFTs suitable for Exascale computers
with heterogeneous architectures.
It implements a parallel algorithm similar to PFFT, with several all-to-all communication steps, and it has capabilities for  CPUs as well as GPUs. 
The heFFTe package can handle various input and output distributions, such as slab, pencil, and 3D blocks (called \emph{bricks}),
and internally it redistributes the data by an all-to-all step which the authors call \emph{tensor transposition}.

\subsection{Contributions}
This article provides the following contributions:
\begin{itemize}
\item We present a parallel multidimensional FFT  algorithm based on the cyclic distribution
that (i) has only a single all-to-all communication step; (ii) works for up to $p_{\max} = \sqrt{N}$ processors,
where $N$ is the total number of array elements; (iii) starts and finishes in the same distribution.
\item We present the program FFTU which is an implementation of this algorithm for an arbitrary dimension $d$
and present timing results for $d = 2,3,5$, demonstrating state-of-the-art performance by comparing
        to FFTW, PFFT, and heFFTe, and showing scalability for higher dimensions.
\end{itemize}

Starting and ending in the same distribution has the following advantages.
First, the same program can be used for the forward and the inverse FFT,
just with the weights conjugated and the outcome scaled by $1/N$, instead of reversing all the steps of the FFT.
Furthermore, the inverse FFT can be performed directly after the FFT,
possibly with an elementwise local operation interspersed, without the need
for data reordering.

\section{Generalized four-step framework}
In this section, we first present the sequential four-step framework for the 1D FFT and derive from this a parallel cyclic-to-cyclic 1D FFT algorithm.
We then generalize this algorithm to the multidimensional case and prove that the resulting parallel multidimensional FFT
algorithm does what it is supposed to do.

\subsection{One-dimensional algorithms}

Suppose that $x$ is an array of length $n$, that $p | n$ and $p | \frac{n}{p}$ (or equivalently $p^2 | n$), and that we want to calculate $y = F_n(x)$. We write $v(a:b:c)$ to mean the strided subarray of $v$ which starts at index $a$ and has stride $b$. The last variable $c$ is the length of the whole array $v$. If we write $v(a:b)$, we mean the subarray that starts at $a$ and ends at $b$ (inclusive).

To establish our notation and lay the basis for generalization to the parallel and multidimensional case, we will first derive our basic sequential algorithm, the so-called \emph{four-step framework}.

\subsubsection{Sequential algorithm}

As any $0 \leq j < n$ can be written uniquely in the form $kp + s$ with $0 \leq k < \frac{n}{p}$, $0 \leq s < p$, we have

\begin{align}\label{eq:radix}
y_a = \sum_{j = 0}^{n - 1} x_j \omega_n^{ja} = \sum_{s = 0}^{p - 1} \sum_{k = 0}^{\frac{n}{p}-1} x_{kp + s} \omega_n^{(kp + s)a} = \\
\sum_{s = 0}^{p - 1} \omega_n^{sa} \sum_{k = 0}^{\frac{n}{p} - 1} x_{kp + s} (\omega_n^{p})^{ka} = \sum_{s = 0}^{p - 1} \omega_n^{sa} \sum_{k = 0}^{\frac{n}{p} - 1} x_{kp + s} \omega_{\frac{n}{p}}^{ka}. \nonumber
\end{align}

Write $x^{(s)}$ for $x(s: p : n)$ and $z^{(s)}$ for $F_{\frac{n}{p}}(x(s:p:n))$. We have the correspondence $x^{(s)}_{k} = x_{kp + s}$ and $x_j = x^{(j \md p)}_{j \divi p}$ and likewise for $z$, $z^{(s)}$. Here, the div operator is defined by
$j \divi p =\lfloor j/p \rfloor$. By periodicity we have $\omega_{\frac{n}{p}}^{ka} = \omega_{\frac{n}{p}}^{k(a \md \frac{n}{p})}$. Armed with this knowledge we can state
\begin{equation}
\sum_{k = 0}^{\frac{n}{p} - 1} x_{kp + s} \omega_{\frac{n}{p}}^{ka} = \sum_{k = 0}^{\frac{n}{p} - 1} x^{(s)}_k \omega_{\frac{n}{p}}^{k(a \md \frac{n}{p})}.
\end{equation}

But this is the definition of the $(a \md \frac{n}{p})$th entry of the DFT of $x^{(s)}$, so we can write (\ref{eq:radix}) as
\begin{equation}
y_a = \sum_{s = 0}^{p - 1} \omega_n^{sa} z^{(s)}_{a \md \frac{n}{p}} ,
\end{equation}
which looks suspiciously much like an FFT of length $p$. In order to massage this equation into that form, we write $a = t\frac{n}{p} + k$ for $0 \leq k < \frac{n}{p}$, $0 \leq t < p$. This yields
\begin{equation}
y_{t\frac{n}{p} + k} = \sum_{s = 0}^{p - 1} \omega_n^{st\frac{n}{p} + sk} z^{(s)}_{k} = \sum_{s = 0}^{p - 1} \omega_p^{st}\left( z^{(s)}_{k} \omega_n^{sk} \right).
\end{equation}
Fixing $k$, we see that $y(k: \frac{n}{p} : n)=F_p(w^{(k)})$, where $w^{(k)}_s = z^{(s)}_k \cdot \omega_n^{sk}$ for $0 \leq s <p$.

We can summarize this as Algorithm~\ref{algo:four_step_sequential}, which is known as  the four-step framework and can be found with a different proof and formulation in \cite{vanloan92}. Step 1 is  called \emph{twiddling} in \cite{vanloan92} and elsewhere. We can easily avoid using extra arrays $z$ and $w$ by reusing $x, y$. The only step where it is not immediately obvious how to do this is Step 3. Here, we note that  $\{y(k : \frac{n}{p} : n) \ | \ 0 \leq k < \frac{n}{p}\}$ splits  $y$ into $\frac{n}{p}$ strided subarrays of length $p$ each, so we can use those to store the $\frac{n}{p}$ arrays $w^{(k)}$ of length $p$.
Therefore, Algorithm~\ref{algo:four_step_sequential} can be executed mostly in-place, where only in Step 2 special care has to be taken
to limit the amount of extra memory needed to permute $z$ into $w$.

\begin{algorithm}
\caption{Sequential four-step framework}\label{algo:four_step_sequential}
Input: $x: $ array of length $n$, number $p$ such that $p^2 | n$. \\
Output: $y: $ array of length $n$, such that $y = F_n(x)$.
\begin{algorithmic}[1]
	\For{$s := 0$ \textbf{ to } $p - 1$} \Comment{Step 0}
		\State $x^{(s)} := x(s: p: n)$;
		\State $z^{(s)} := F_{\frac{n}{p}}(x^{(s)})$;
	\EndFor

	\For{$s := 0$ \textbf{ to } $p - 1$} \Comment{Step 1}
		\For{$k := 0$ \textbf{ to } $\frac{n}{p} - 1$}
			\State $z^{(s)}_k := \omega_n^{ks} z^{(s)}_k$;
		\EndFor
	\EndFor

	\For{$s := 0$ \textbf{ to } $p - 1$} \Comment{Step 2}
		\For{$k := 0$ \textbf{ to } $\frac{n}{p}-1$}
			\State $w^{(k)}_s := z^{(s)}_k$;
		\EndFor
	\EndFor

	\For{$k := 0$ \textbf{ to } $\frac{n}{p} - 1$} \Comment{Step 3}
		\State $y(k: \frac{n}{p}: n) := F_{p}(w^{(k)})$;
	\EndFor
\end{algorithmic}
\end{algorithm}

\subsubsection{Parallel algorithm}

We use the Bulk Synchronous Parallel (BSP) model \cite{valiant90} as our approach to the parallelization of the FFT.
BSP algorithms comprise a sequence of \emph{supersteps}, each performing a computation stage or a communication stage, where each superstep is terminated by a global synchronization.
Following the BSP approach in~\cite{bisseling20},
we express our parallel algorithms in Single Program Multiple Data (SPMD) style, by giving the
program text for processor $P(s)$ with $s \in [p]$. This means that although the executed program depends on $s$, we do not have to write different texts for different processors.
We express the communication in our algorithm by the one-sided primitive ``Put," which completely defines
the transferal of data by an action from the sender. The sender $P(s)$ states to which destination processor
the data has to be sent, and to which memory address on that processor. The data can then be used in the next superstep.
There is no need to define actions by receivers.

We will analyze the time complexity of parallel algorithms within the BSP framework.
The cost of a computation superstep in the BSP model depends on the maximum number of floating-point operations (flops)
that a processor carries out, and we will measure these in real flops, as is usual.
The cost of a communication superstep depends on the maximum number of data words
each processor sends or receives, here complex numbers, and we will take $g$ (for \emph{gap}) as the cost per data word.
The cost of the synchronization at the end of a superstep is taken as $l$ (for \emph{latency}),
but we will charge this cost only for communication supersteps.
(The reason is that we only use Put operations in our program, not Get operations, and this means
that in an actual implementation the synchronization at the end of a computation superstep
can be removed.) For more details on BSP, and an extensive treatment of the BSP 1D FFT
with the cyclic distribution, see~\cite{bisseling20}.

We will now derive a parallel version of the four-step framework. Steps 0 and 1 of Algorithm~\ref{algo:four_step_sequential} are trivial to parallelize, and they determine the distribution: namely, cyclic over $p$ processors.  Here, processor $P(s)$ simply performs iteration $s$ of the outer loops.
This yields Superstep~0  of the parallel 1D algorithm, Algorithm~\ref{algo:four_step_parallel}.
We would like to end in the same cyclic distribution in which we started.
Let $y^{(s)} := y(s: p: n)$ be the part of the output vector $y$ stored by the cyclic distribution on $P(s)$.
Component $c$ of this part is then denoted by $y^{(s)}[c] := y(s: p: n)[c]$.

We will first parallelize the final step of the sequential algorithm, Step 3, aiming at a completely local operation in the cyclic
distribution. Note that in the cyclic distribution over $p$ processors,
global component $y_j$ is stored on processor $P(j \md p)$ in position $j \divi p$.
Consider component $c$ of $y(k : \frac{n}{p} : n)$, the output of the final step. Because $p | \frac{n}{p}$, we have that component $y(k : \frac{n}{p} : n)[c] = y_{c\frac{n}{p} + k}$ is stored on $P((c\frac{n}{p} + k) \md p) = P(k \md p)$ in local position $(c\frac{n}{p} + k) \divi p = c\frac{n}{p^2} + k \divi p$. So $y(k: \frac{n}{p} : n)$ is stored in its entirety on $P(k \md p)$.
Taking $c=0$, we see that  $y(k: \frac{n}{p}: n)[0]$ is stored in local position $k \divi p$. Furthermore,
incrementing $c$ increases the local position by $\frac{n}{p^2}$. Therefore, $y(k: \frac{n}{p} : n) = y^{(k \md p)}(k \divi p: \frac{n}{p^2}: \frac{n}{p})$, which is local on $P(s)$ for $ k \md p = s$. As $k$ ranges from $0$ to $\frac{n}{p} - 1$, we have that $t = k \divi p$ ranges from $0$ to $\frac{n}{p^2} - 1$. This determines the $k$-values for which the output of Step 3 is local,
and we require the input also to be local. This yields Superstep~2, a local transformation of  $y^{(s)}(t: \frac{n}{p^2}: \frac{n}{p})$ by $F_p$ for $0 \leq t < \frac{n}{p^2}$.

Finally, we parallelize Step 2 of the four-step framework, which is a permutation of the data,
and is usually called the \emph{transposition} step. This step translates into a communication superstep that should transfer the  components of the vector $w$ to their desired destination, so that the following superstep becomes local.
Taking $c=s$ in the above, we obtain $y(k: \frac{n}{p} : n)[s] = y^{(k \md p)}(s\frac{n}{p^2} + k \divi p)$.
This means that we need to send the $x^{(s)}_k$ produced in Superstep~0 to $P(k \md p)$ in local position $s\frac{n}{p^2} + k \divi p$. So $x^{(s)}(k : p : \frac{n}{p})$ ends up in $P(k)$ for $0 \leq k < p$. The first element corresponds to $y^{(k)}(s\frac{n}{p^2})$, and every time we increase $k$ by $p$, we increment $k \divi p$, so $x^{(s)}(k : p : \frac{n}{p})$ corresponds to $y^{(k)}(s\frac{n}{p^2}: (s + 1)\frac{n}{p^2} - 1)$. This is achieved by the Put-statement of Superstep~1.
This completes our algorithm. If desired, we can make the algorithm completely in-place by taking $y = x$.

\begin{algorithm}
\caption{Parallel four-step framework for $P(s)$}\label{algo:four_step_parallel}
Input: $x: $ array of length $n$, distr$(x) = $ cyclic over $p$ processors such that $p^2 | n$. \\
Output: $y: $ array of length $n$, distr$(y) = $ cyclic, such that $y = F_n(x)$.
\begin{algorithmic}[1]
	\State $x^{(s)} := F_{\frac{n}{p}}(x^{(s)})$; \Comment{Superstep 0}
	\For{$k := 0$ \textbf{ to } $\frac{n}{p} - 1$}
			\State $x^{(s)}_k := \omega_n^{ks} x^{(s)}_k$;
	\EndFor

	\For{$k := 0$ \textbf{ to } $p - 1$}\Comment{Superstep 1}
		\State Put $x^{(s)}(k : p : \frac{n}{p})$ in $P(k)$ as $y^{(k)}(s\frac{n}{p^2}: (s + 1)\frac{n}{p^2} - 1)$;
	\EndFor

	\For{$t := 0$ \textbf{ to } $\frac{n}{p^2} - 1$}\Comment{Superstep 2}
		\State $y^{(s)}(t : \frac{n}{p^2} : \frac{n}{p}) := F_p(y^{(s)}(t : \frac{n}{p^2} : \frac{n}{p}))$;
	\EndFor
\end{algorithmic}
\end{algorithm}

\subsection{Parallel multidimensional  algorithm}\label{section_multialgorithm}

We will generalize Algorithm~\ref{algo:four_step_parallel} directly to a parallel multidimensional
algorithm, Algorithm~\ref{algo:dd_4-step}, which we prove correct by showing that it computes the multidimensional DFT of (\ref{eq:dftd}). This equation can
concisely be formulated in tensor notation (see~\cite{vanloan92}) as
\begin{equation}
Y = (F_{n_1} \otimes \cdots \otimes F_{n_d})(X).
\end{equation}
We use vector notation to abbreviate dimensionwise strides, subarrays, etc. So we define, for instance,
$X(t:\frac{n}{p^2}:\frac{n}{p}) = X(t_1:n_1/p_1^2:n_1/p_1, \ldots, t_d:n_d/p_d^2:n_d/p_d) $.
The algorithm can be performed in-place using only an array $X$,
but we have used different variables $X, Y, Z, W, V$ to facilitate the proof.

It is no coincidence that
Algorithm~\ref{algo:dd_4-step} combines instances of Algorithm~\ref{algo:four_step_parallel} dimensionwise.
This happens because the multidimensional Fourier transform is the tensor product of 1D Fourier transforms
and because we can combine BSP algorithms for linear functions into a BSP algorithm for the tensor product,
as shown in~\cite{koopman22}. The proof of this powerful observation involves category theory.
Our correctness proof of Algorithm~\ref{algo:dd_4-step}, however, requires no such category theory, and is
based on basic algebraic manipulations and concise notation.

\begin{algorithm}
\caption{Parallel four-step framework for processor $P(s)=P(s_1,\ldots,s_d)$ in $d$ dimensions}\label{algo:dd_4-step}
Input: $X: $ multidimensional array of size $n_1 \times \cdots \times n_d$, distr$(X) = $ $d$-dimensional cyclic over $p_1 \times \cdots \times p_d$ processors such that $p_l^2 | n_l$, for $l=1,\ldots,d$. \\
Output: $V: $ multidimensional array of size $n_1 \times \cdots \times n_d$, distr$(V) = $  $d$-dimensional cyclic, such that $V = (F_{n_1} \otimes \cdots \otimes F_{n_d})(X)$.
\begin{algorithmic}[1]
	\State $Y^{(s)} := (F_{n_1/p_1} \otimes \cdots \otimes F_{n_d/p_d})(X^{(s)})$; \Comment{Superstep 0}
	\For{$k \in [n_1/p_1] \times \cdots \times [n_d/p_d]$}
		\State $Z^{(s)}[k] := (\prod_{l = 1}^d\omega_{n_l}^{k_ls_l}) Y^{(s)}[k]$;
	\EndFor
	\For{$k \in [p_1] \times \cdots \times [p_d]$}\Comment{Superstep 1}
		\State Put $Z^{(s)}(k : p : \frac{n}{p})$ in $P(k)$ as  $W^{(k)}[s\frac{n}{p^2}: (s + 1)\frac{n}{p^2} - 1]$;
	\EndFor
	\For{$t \in [n_1/p_1^2] \times \cdots \times [n_d/p_d^2]$} \Comment{Superstep 2}
		\State $V^{(s)}(t:\frac{n}{p^2}: \frac{n}{p}) := (F_{p_1} \otimes \cdots \otimes F_{p_d})\left(W^{(s)}(t:\frac{n}{p^2}: \frac{n}{p})\right)$;
	\EndFor

\end{algorithmic}
\end{algorithm}

\begin{thm}
Algorithm~\ref{algo:dd_4-step} computes the multidimensional DFT expressed in (\ref{eq:dftd}).
\end{thm}
\begin{proof}
Applying the definition (\ref{eq:dftd}) of the multidimensional DFT for the local array of $P(s)$ gives
\begin{equation}
Y^{(s)}[k] = \sum_{j_1 \in [n_1/p_1] } \cdots \sum_{j_d \in [n_d/p_d ]} X^{(s)}[j] \omega_{n_1/p_1}^{j_1k_1} \cdots \omega_{n_d/p_d}^{j_dk_d}.
\end{equation}

By using  $\omega_{n_l/p_l}^{j_lk_l} = \omega_{n_l}^{j_lk_lp_l}$ and twiddling $Y^{(s)}$, we obtain
\begin{equation}
Z^{(s)}[k] = \sum_{j_1 \in [n_1/p_1]} \cdots \sum_{j_d \in [n_d/p_d]} X^{(s)}[j] \omega_{n_1}^{k_1(j_1p_1 + s_1)} \cdots \omega_{n_d}^{k_d(j_dp_d + s_d)},
\end{equation}
which is the output of Superstep~0.

Now consider component $s'$ of the input array $W^{(s)}(t:\frac{n}{p^2}:\frac{n}{p})$ of Superstep~2.
This component corresponds to $W^{(s)}[t + s'\frac{n}{p^2}] = W^{(s)}(s'\frac{n}{p^2}: (s' + 1)\frac{n}{p^2} - 1)[t]$,
which has been obtained by communication in Superstep~1
of component $Z^{(s')}(s:p:\frac{n}{p})[t] = Z^{(s')}[s + tp]$. This means that for $P(s)$ in the role of receiver of a block of data,
$P(s')$ is the sender.
Therefore,
\begin{multline}
W^{(s)}(t:\frac{n}{p^2}:\frac{n}{p})[s'] =
Z^{(s')}[s + tp] \nonumber \\
 = \sum_{j_1 \in [n_1/p_1]} \cdots \sum_{j_d \in [n_d/p_d]} X^{(s')}[j] \,\omega_{n_1}^{(s_1 + t_1p_1)(j_1p_1 + s_1')} \cdots \omega_{n_d}^{(s_d + t_dp_d)(j_dp_d + s_d')}.
\end{multline}

The final output $ V^{(s)}(t:\frac{n}{p^2}:\frac{n}{p})$ of Algorithm~\ref{algo:dd_4-step}
is computed in Superstep~2.
Consider component $k$ of the output array $ V^{(s)}(t:\frac{n}{p^2}:\frac{n}{p})$,
which can be written as $V^{(s)}[t+ k\frac{n}{p^2} ]$ and is computed by a local DFT,
\begin{multline}
V^{(s)}[t+ k\frac{n}{p^2} ]   = \sum_{s_1' \in [p_1]} \cdots \sum_{s_d'  \in [p_d]}
W^{(s)}(t:\frac{n}{p^2}:\frac{n}{p})[s'] \, \omega_{p_1}^{s_1'k_1} \cdots \omega_{p_d}^{s_d'k_d} \\
\quad = \sum_{s_1' \in [p_1]} \cdots \sum_{s_d'  \in [p_d]} \
\sum_{j_1 \in [n_1/p_1]} \cdots \sum_{j_d \in [n_d/p_d]} X^{(s')}[j] \,
\omega_{n_1}^{(s_1 + t_1p_1)(j_1p_1 + s_1')} \\
\cdots \omega_{n_d}^{(s_d + t_dp_d)(j_dp_d + s_d')}
\omega_{p_1}^{s_1'k_1} \cdots \omega_{p_d}^{s_d'k_d}.
\end{multline}

To work toward the definition of the multidimensional DFT, we rewrite the powers of the roots of unity
using $\omega_{p_l}^{s_l'k_l}  = \omega_{n_l}^{s_l'k_ln_l/p_l} = \omega_{n_l}^{(j_lp_l+s_l')k_ln_l/p_l}$.
Note that  the extra factor $\omega_{n_l}^{j_lp_lk_ln_l/p_l} = \left(\omega_{n_l}^{n_l}\right)^{j_lk_l} = 1$.
This gives the expression
\begin{multline}
\sum_{s_1' \in [p_1]} \cdots \sum_{s_d'  \in [p_d]}
\sum_{j_1 \in [n_1/p_1]} \cdots \sum_{j_d \in [n_d/p_d]} X^{(s')}[j]
\omega_{n_1}^{(s_1 + t_1p_1+ k_1n_1/p_1)(j_1p_1 + s_1')} \cdots \\
\omega_{n_d}^{(s_d + t_dp_d+ k_dn_d/p_d)(j_dp_d + s_d')}.
\end{multline}
Because we can write any number $j_l' \in [n_l]$  uniquely as $j_l' = s_l' + j_lp_l$ for $0 \leq s_l < p_l$, $0 \leq j_l < n_l/p_l$ and because of the local-global relations $X^{(s')}[j] = X[s' + jp]$ and
$V^{(s)}[t+ k\frac{n}{p^2} ] = V[s+tp+k\frac{n}{p}]$, we obtain
\begin{multline}
V [ s_1+t_1p_1+k_1\frac{n_1}{p_1}, \ldots , s_d+t_dp_d+k_d\frac{n_d}{p_d}]
= \\
\sum_{j_1' \in [n_1]} \cdots \sum_{j_d' \in [n_d]} X[j'] \,
\omega_{n_1}^{(s_1 + t_1p_1+ k_1n_1/p_1)j_1'} \cdots \omega_{n_d}^{(s_d + t_dp_d+ k_dn_d/p_d)j_d'}.
\end{multline}
This proves that the corresponding global array $V$ is precisely the image of the global array $X$ under the $d$-dimensional DFT.
\end{proof}

\subsection{Complexity analysis}

We analyze the time complexity of our parallel multidimensional FFT algorithm within the
BSP model~\cite{valiant90}.
Assume we have an array of size $n_1 \times \cdots \times n_d$.
Let $N = n_1 \cdots n_d$ be the total number of array elements and $p=p_1 \cdots p_d$
be the total number of processors. (Note that we slightly abuse the notation,
since we have used the variable $p$ in the previous section in a different meaning, namely as a tuple.)
The sequential multidimensional FFT performs about $5 N \log N$ flops
for $N$ array elements, regardless of dimension or shape of the array.
As the number of operations per array element  is small, CPU--RAM bandwidth and cache architecture matter
more than the precise number of flops, but we will disregard this in our theoretical analysis
and consider it part of the implementation concerns.

Let us first look at the computation supersteps of Algorithm~\ref{algo:dd_4-step}. The first part of Superstep~0 takes $5  \frac{N}{p} \log \frac{N}{p}$ flops, and Superstep~2 takes $\frac{N}{p^2} \cdot 5p \log p$.
With proper implementation, the twiddling takes about two complex multiplications per element, as we will see
in Algorithm~\ref{algo:pack_and_twiddle} in the next section. So this adds $12\frac{N}{p}$ real flops,
giving a total number of flops of
\begin{equation}
T_{\mathrm{comp,FFT}} = 5\frac{N}{p} \log\frac{N}{p} + 5\frac{N}{p}\log p + 12\frac{N}{p}
= 5 \frac{N}{p}  \log N + 12\frac{N}{p}.
\end{equation}
Except for the additional $12\frac{N}{p}$ flops from the twiddling, the original
sequential workload is perfectly divided into $p$ parts.

Looking at the communication superstep, Superstep~1, we note that
each array element is communicated at most once, with each processor sending and receiving at most
$\frac{N}{p}$ elements, at a cost of  $\frac{N}{p} g$. This is a balanced all-to-all communication step.
Since we have only one communication superstep for which we charge the synchronization,
the total synchronization cost of the algorithm is $l$.
Therefore, the total BSP cost of our  parallel multidimensional algorithm is
\begin{equation}
T_{\mathrm{FFT}} =
5 \frac{N}{p} \log N + 12 \frac{N}{p} + \frac{N}{p} g + l.
\end{equation}

Our algorithm has a scalability limit for the maximum number of processors $p_{\max}$ that can be used,
because in dimension $l$ we can use up to $\sqrt{n_l}$ processors;
this is in the ideal case that $p_l^2 | n_l$.
If all $n_l$ are squares, this limit equals
\begin{equation}
p_{\max}   = (n_1 \cdots n_d)^{1/2} = \sqrt{N}.
\end{equation}
If not all $n_l$ are squares, the maximum is a bit lower. In case the $n_l$ and $p_l$ are powers of 2,
$p_{\max}$ is then lowered by a factor of 2 for every $n_l$ that is not a power of 4.
For a 3D array of size $1024^3$, our algorithm can use up to $32^3  =$\,32,768 processors.
For 3D arrays of size $256^3$ and $512^3$, our algorithm can use up to $16^3  =$\,4096 processors.
For a 2D array of size $2^{24} \times 64$, with the same number of elements as the array of size $1024^3$
but with a very high aspect ratio, we can still use the same number of processors,  $p_{\max}=$\,32,768.

It is possible to scale beyond $p_{\max} =\sqrt{N}$, but in that case more than one communication superstep is needed
and a generalization of the cyclic distribution must be used, called the group-cyclic distribution~\cite{inda01};
see also~\cite{bisseling20} for an explanation and an implementation for the 1D FFT.
The  \emph{group-cyclic distribution}
with cycle $c$ splits an array of size $n$ into $\frac{p}{c}$ blocks
of size $\frac{cn}{p}$. It assigns each block
to a group of $c$ processors, and it uses the cyclic distribution within each block.
In this distribution, array element $x_j$ is assigned to processor $P(( j \divi \frac{cn}{p})c + j\md c)$.
This is a different generalization of the cyclic distribution from the
well-known block-cyclic distribution, which is used in certain parallel numerical linear algebra packages,
for instance, in ScaLAPACK~\cite{blackford97}.

\section{Implementation}

In this section, we present the implementation of our multidimensional FFT algorithm, which we humbly call FFTU, the fastest FFT in Utrecht. The implementation can be found at \url{https://gitlab.com/Thomas637/FFT} and it contains both an MPI implementation and a BSPlib implementation.
FFTU has been released under the GNU GPL license. We discuss only the MPI implementation since the BSPlib version is very similar. The only difference is that in the BSPlib version we unpack the data packets manually instead of letting the communication library do this.

We use FFTW for the local computations and MPI for communication. We move the data into a buffer so that the data to be sent to a single processor is stored contiguously; this is called \emph{packing}. We have implemented a method where we send the packets using \texttt{MPI\_alltoall} and then locally rearrange the data, moving them into the correct positions (called \emph{unpacking}), and also a method using \texttt{MPI\_alltoallv}, where we specify how the packet is to be stored on the remote processor using derived data types. This allows (but does not require) the MPI implementation to perform the communication without local data movement.

We combine the packing with the twiddling to minimize the consumption of CPU--RAM bandwidth.
This gives Algorithm~\ref{algo:pack_and_twiddle}, which can be applied to the local array $X^{(s)}[k]$.
Note that we perform only two complex multiplications per data
element in the innermost loop of the algorithm, so that the time complexity is about $12N/p$ real flops.

The twiddle weights $\omega_{n_l}^{k_ls_l} $ used in the FFT algorithm can be precomputed
and stored in a weight table, which requires a local memory of
\begin{equation}
\label{eq:weights}
M_{\mathrm{twiddle}} = \sum_{l=1}^d \frac{n_l}{p_l},
\end{equation}
which is much less than the  $\prod_{l=1}^d \frac{n_l}{p_l} = \frac{N}{p}$ memory needed to store the local array $X^{(s)}$.

\begin{algorithm}[ht]
\caption{Packing and twiddling}\label{algo:pack_and_twiddle}
Input: a $d$-dimensional array $X$ of size $n_1/p_1 \times \cdots \times n_{d}/p_{d}$ in row-major format,
$s = (s_1,\ldots,s_d) \in [p]$. \\
Output: $d$-dimensional arrays packet$_k$ containing the twiddled $X(k:p:\frac{n}{p})$ in row-major format.
\begin{algorithmic}[1]
	\For{$t_1 := 0$ \textbf{ to } $n_1/p_1 - 1$}
		\State factor$_1 := \omega_{n_1}^{t_1s_1}$;
		\For{$t_2 := 0$ \textbf{ to } $n_2/p_2 - 1$}
			\State factor$_2 := \text{factor}_1\cdot \omega_{n_2}^{t_2s_2}$
			\State $\vdots$
			\For{$t_{d} := 0$ \textbf{ to } $n_{d}/p_{d} - 1$}
				\State factor$_{d} := \text{factor}_{d - 1} \cdot \omega_{n_{d}}^{t_{d} s_{d}}$;
				\State packet$_{t \md p}[t \divi p] := X[t] \cdot \text{factor}_{d}$;
			\EndFor
		\EndFor
	\EndFor
\end{algorithmic}
\end{algorithm}

For the sequential multidimensional FFTs that we use in Algorithm~\ref{algo:dd_4-step}, numerous optimized libraries are available, such as FFTW~\cite{frigo05} and SPIRAL~\cite{puschel05}. We have chosen FFTW as it has the most flexible interface, which we exploit for the interleaved strided arrays, but our algorithm can also easily be adapted to use SPIRAL or other libraries. The row-major format of the input and output means
that the last dimension is consecutive in memory, like in the programming language C.

\section{Numerical experiments}

\subsection{Setup}

We performed strong scaling experiments for four different programs: our program FFTU, FFTW, PFFT, heFFTe. The first three are compiled with Intel 2021.2.0,
which includes Intel MPI, and flags -O3 -march=native unless otherwise specified. The FFTW version is 3.3.9. Finally we use PFFT commit e4cfcf9 on \url{https://github.com/mpip/pfft} (no version number available). For $p = 1, 2$, the derived data type version of FFTU with Intel MPI does not terminate. For this reason, we used the manual unpacking method for $p = 1, 2$. We use heFFTe version 2.2 with Intel's Math Kernel Library version 2022.1.0.

FFTW can generate faster FFT functions at the cost of some setup time. This is controlled by a flag, offering choices \texttt{FFTW\_ESTIMATE}, \texttt{FFTW\_MEASURE}, \texttt{FFTW\_PATIENT}. We tested these locally on an array of size $256 \times 256 \times 256$ and the execution (setup) time was $2.331$ ($0.03$), $0.176$ ($2.735$), $0.170$ ($239$) seconds, respectively. \linebreak As \texttt{FFTW\_PATIENT} only pays off after about $40,000$ executions, we chose to use
\linebreak \texttt{FFTW\_MEASURE}. In heFFTe, the use of \texttt{FFTW\_ESTIMATE} is hard coded, and replacing this by \texttt{FFTW\_MEASURE} leads to runtime errors. For this reason, we use Intel's Math Kernel Library, which is also the sequential time reported for heFFTe.

We ran our timing experiments on arrays with a total number of $N= 2^{30}$ elements,
in shape $1024^3$ for all four programs, in shape $64^5$ for FFTU, PFFT, and FFTW,
and in shape 16,777,216\,$ \times$\,64 for FFTU and FFTW.
We obtained timings
on the thin node partition of the supercomputer Snellius at SURFsara in Amsterdam. A thin node consists of two 64-core AMD Rome 7H12 processors in a dual socket configuration, running at 2.6 GHz. There is 2 GiB of RAM available per core. The interconnect is Infiniband HDR100 (100 Gbps) with a fat tree topology. The job scheduler used is SLURM, and we do not share nodes. The operating system is CentOS7. We use a single switch in our experiments and bound the MPI processes to cores with \texttt{--cpu\_bind=cores}.

Ideally, we would write down the time at the start of the FFT and at the end, and subtract the two values to obtain a timing result. The problem with this method is that the synchronizing function \texttt{MPI\_Barrier} of the MPI library only guarantees that no processor leaves the barrier before all processors have entered. So there is no guarantee that every processor starts the FFT function at the same time. That is why we apply the FFT 100 times, so that the small difference in leaving the barrier becomes negligible. For heFFTe, we use the benchmark program \texttt{speed3d\_c2c} provided by the heFFTe package.

For FFTW and PFFT, we measure both the time it takes when we demand that the programs end in the same distribution that we started in
(default for  FFTW and with flag PFFT\_TRANSPOSED\_NONE for PFFT), and when we allow them to end in a different distribution (flags FFTW\_TRANSPOSED\_OUT  and PFFT\_TRANSPOSED\_OUT, respectively). HeFFTe does not provide an explicit option to end in the same distribution.

\subsection{Results}

Table~\ref{table:3d} presents the timing results for an FFT of an array of size $1024^3$ for different numbers of processors $p$.
This represents the 3D problem, which is likely the most common application of a multidimensional FFT,
used in modeling our 3D physical space.
If we require the input and output to be in the same distribution, we see that FFTU performs better than PFFT in all cases,
and better than FFTW for $p \geq 128$. For smaller $p$, FFTW outperforms both PFFT and FFTU, which invoke FFTW
but entail additional overhead. This is particularly clear from the large parallel overhead for $p=1$ when compared to
the sequential case, which represents a slowdown of a factor 2.3 for FFTU and 2.9 for PFFT.
Once we exceed the number of cores in a socket, i.e., for $p>64$, communication becomes more costly
and the advantage of fewer communication supersteps becomes clear.
The highest speedup achieved for
FFTU compared to sequential FFTW is 149$\times$ on  $4096$ processors, and for PFFT it is  98.5$\times$. FFTW can use only 1024 processors, with a speedup of 32.1$\times$. heFFTe achieves a speedup of 119$\times$ on  $4096$ processors, relative to Intel's Math Kernel Library.
Counting $5 N \log N $ flops to transform an array of $N$ elements,
FFTU reaches a top computing rate of $0.946$ Tflop/s for this problem.

If we allow the input and output distribution to be different, both PFFT and FFTW can save the final
communication superstep. For this 3D problem, all three programs
perform a single communication superstep for $p \leq 1024$, and FFTU does this for all given $p$,
because it could go up to $p=$\,32,768 with a single communication superstep.
For $p>1024$, PFFT needs $\lceil \frac{3}{3 - 2} \rceil - 1 = 2$ communication supersteps, because it then uses a 2D decomposition, putting as many processors along the first dimension as possible.
For FFTW,  the maximum number of processors that can be used is 1024.

PFFT with output allowed in a different distribution is about equally fast as FFTU.
It is surprising that this also holds for $p = 2048, 4096$, given the extra superstep PFFT needs.
When inspecting the instance $p = 4096$ for the case with the same distribution imposed,
we noticed that the final communication superstep performing the redistribution
takes as long as the entire calculation,
where we would expect it to be less than half.
This means that PFFT performs its internal two communication supersteps very efficiently.
It does this by using the global transposes of FFTW instead of directly calling MPI,
and optimizing these by precomputing several FFTW plans; see~\cite[Section 3.3]{pippig13}.
Furthermore, we observe a superlinear speedup of PFFT for $p=4096$ compared with $p = 2048$.

In Table~\ref{table:3d}, we see that FFTW with output allowed in a different distribution is faster than FFTU. This can be explained by FFTW doing a better job at communicating in its communication superstep, perhaps by better exploiting the shared memory available. We performed an additional experiment to see whether we could improve
the performance of FFTU. Using OpenMPI version 4.1.1 instead of Intel MPI,  the time needed
decreased from $0.664$ s to $0.515$ s for $p=512$, actually even beating FFTW.
So there is room for improvement in the MPI implementation
and it might be possible to attain the same performance as FFTW.

Although we cannot accurately compare the absolute running times of FFTU with heFFTe due to the different sequential library, 
we can compare the way they both scale. We see similar good scaling for FFTU and heFFTe, even when going to $2048$ processors, where an extra communication step is needed. Like PFFT, heFFTe can hide this communication well. 
We see a similar speedup for FFTU and heFFTe when going to $4096$ processors, 
but in both cases we do not observe the superlinear speedup as shown by PFFT.

\begin{table}[ht]
\caption{Time (in s) of a multidimensional FFT for a $1024^3 $ array using four  programs: FFTU, PFFT, FFTW, and heFFTe.
The program FFTU yields the same data distribution for the output as for the input. For PFFT and FFTW, this can be imposed (``same")
or not (``different"). For comparison, the  time for running the sequential FFTW is also given.}
\begin{center}
\begin{tabular}{rrrrrrr}
\hline
\multicolumn{1}{c}{$p$}& \multicolumn{1}{c}{FFTU}&
    \multicolumn{2}{c}{PFFT} & \multicolumn{2}{c}{FFTW} & \multicolumn{1}{c}{heFFTe} \\
    & \multicolumn{1}{c}{same} & \multicolumn{1}{c}{same} & different & \multicolumn{1}{c}{same} & different & \multicolumn{1}{c}{different}\\
\hline
    seq &&&&17.541& & 32.834 \\
    1 & 40.065& 51.334& 21.646& 23.025& 19.615& -\\
    2 & 18.058& 27.562& 12.359& 13.650& 12.519& 18.385 \\
    4 & 8.074& 13.179 & 6.432& 6.962& 6.236& 15.354 \\
    8 & 3.999& 9.102& 4.290& 4.024& 3.260& 8.167 \\
    16 & 2.349& 5.552& 2.510& 2.388& 1.803& 5.409 \\
    32 & 1.789& 3.190& 1.417& 1.545& 1.145& 3.589 \\
    64 & 1.802& 3.133& 1.411& 1.670& 1.378& 2.814 \\
    128 & 1.366& 3.330& 1.461& 1.996& 1.475& 2.782 \\
    256 & 0.980& 1.972& 0.918& 1.208& 0.797& 1.905 \\
    512 & 0.664& 1.409& 0.677& 0.991& 0.577& 1.236 \\
    1024 & 0.317& 0.644& 0.327& 0.546& 0.310& 0.618 \\
    2048 & 0.163& 0.417& 0.223& && 0.393  \\
    4096 & 0.118& 0.178& 0.088& && 0.277 \\
\hline
\end{tabular}
\end{center}
\label{table:3d}
\end{table}

 Table~\ref{table:5d} shows the results for an array of size $64^5$.
 Here, FFTW can only use up to 64 processors,
 so there are no timings for FFTW beyond this number.
Comparing the 5D timings of this table with the 3D timings of Table~\ref{table:3d},  we see
 that higher-dimensional FFTs run faster for the same total number of array elements.
 The highest speedup achieved by FFTU is 176$\times$ and that of PFFT is 225$\times$.

 In Table~\ref{table:5d}, we see that FFTU performs about as well as PFFT
 when we do not require the output to be in the same distribution as the input.
 This makes sense, as a 2D processor distribution suffices, meaning we need $\lceil \frac{5}{5 - 2} \rceil - 1 = 1$ communication superstep. For $p = 128, 512, 1024, 2048$ FFTU is a  bit faster, but for $p = 256, 4096$ this is the other way round. These differences can be explained by the implementation of the communication step. We simply describe to MPI what we want by use of \texttt{MPI\_Alltoallv} and derived data types, and let the black box MPI library handle the exact algorithm for communication. FFTW and PFFT experiment with several algorithms during the planning phase and then choose the fastest option. Apparently, sometimes MPI chooses the superior algorithm, and sometimes FFTW/PFFT does. We again observe superlinear speedup for PFFT when going from $2048$ to $4096$ processors.
 The same happens when doubling the array size to $128 \times 64^4$,
 so we presume that  this phenomenon is not due to cache size.

\begin{table}[ht]
\caption{Time (in s) of a multidimensional FFT for a $64^5 $ array using three  programs: FFTU, PFFT, and FFTW.
The program FFTU yields the same data distribution for the output as for the input. For PFFT and FFTW, this can be imposed (``same")
or not (``different"). For comparison, the  time for running the sequential FFTW is also given.}
\begin{center}
\begin{tabular}{rrrrrr}
\hline
\multicolumn{1}{c}{$p$}& \multicolumn{1}{c}{FFTU}&
\multicolumn{2}{c}{PFFT} & \multicolumn{2}{c}{FFTW} \\
& \multicolumn{1}{c}{same} & \multicolumn{1}{c}{same} & different & \multicolumn{1}{c}{same} & different \\
\hline
seq &&&&17.381& \\
1 & 36.334& 23.981& 16.134& 18.803& 19.451 \\
2 & 17.843& 14.548& 9.844& 12.690& 11.738\\
4 & 7.771& 7.630& 5.053& 6.826& 6.130\\
8 & 4.111& 4.226& 2.746& 3.538& 3.148\\
16 & 2.372& 2.669& 1.614& 2.119& 1.862\\
32 & 1.653& 2.165& 1.125& 1.593& 1.301\\
64 & 1.634& 2.259& 1.222& 1.390& 0.997\\
128& 1.315& 2.735& 1.551&    &    \\
256& 0.965& 1.650& 0.956&    &    \\
512& 0.609& 1.256& 0.667&    &    \\
1024& 0.304& 0.644& 0.357&    &    \\
2048& 0.167& 0.358& 0.190&    &    \\
4096& 0.099& 0.159& 0.077&    &    \\
\hline
\end{tabular}
\end{center}
\label{table:5d}
\end{table}

Table~\ref{table:2d} shows timing results for an array of size 16,777,216\,$ \times $\,64.
This problem with a high aspect ratio failed for PFFT, because of an integer division-by-zero error.
We did not investigate this case further and just present timings for FFTU and FFTW.
Note that both PFFT and FFTW can use at most 64 processors for this problem,
since they need one dimension to be local.

FFTU is slower here than in the 3D and 5D problems with the same $N$. Our explanation is as follows.
 The total size of the twiddle tables is  $\sum_{l = 1}^d \frac{n_l}{p_l}$ (see (\ref{eq:weights})),
 so in the case where one of the $n_l$ is relatively close to $N$ and the others are small,
 this is a relatively large table,
 which cannot be stored in cache and hence leads to a performance penalty.
 The twiddle tables are computed during the algorithm and then repeatedly used;
 we found no need for precomputing the twiddle weights, as is done for the other weights by FFTW.

\begin{table}[ht]
\caption{Time (in s) of a multidimensional FFT for a 16,777,216\,$\times$\,64 array using two programs: FFTU and FFTW.
The program FFTU yields the same data distribution for the output as for the input. For FFTW, this can be imposed (``same")
or not (``different"). For comparison, the  time for running the sequential FFTW is also given.}
\begin{center}
\begin{tabular}{rrrrrr}
\hline
\multicolumn{1}{c}{$p$}& \multicolumn{1}{c}{FFTU}&
 \multicolumn{2}{c}{FFTW} \\
& \multicolumn{1}{c}{same} & \multicolumn{1}{c}{same} & different \\
\hline
seq &&24.182& \\
1 & 43.146& 26.984& 31.440\\
2 & 21.950& 16.661& 17.382\\
4 & 9.613& 8.649& 8.563\\
8 & 5.150& 4.577& 4.609\\
16 & 3.045& 2.695& 2.699\\
32 & 2.347& 2.023& 1.959\\
64 & 2.218& 1.646& 1.442\\
128 & 1.615&  &   \\
256 & 1.264&  &   \\
512 & 0.841&  &   \\
1024 & 0.331&  &   \\
2048 & 0.230&  &   \\
4096 & 0.204&  &   \\
\hline
\end{tabular}
\end{center}
\label{table:2d}
\end{table}

\section{Conclusions}
We have developed an efficient and scalable multidimensional FFT algorithm
that has a single all-to-all communication superstep, that starts and ends in the same distribution,
and that can use up to $\sqrt{N}$ processors (or slightly less if the array size in some dimensions is not a square number).
We have implemented this algorithm in a program FFTU
and tested it on a distributed-memory parallel architecture,
comparing it to three state-of-the-art programs, FFTW, PFFT, and heFFTe.
Our program FFTU uses FFTW for its local computations and it calls MPI for its communication.
We can view our package FFTU as an extension of FFTW,
in the same way as PFFT can be considered an extension.

Overall, FFTU performs on par with PFFT and FFTW if the output distribution need not be the same
as the input distribution, and FFTU outperforms PFFT and FFTW if these distributions must be the same.
There are some fluctuations in the behavior of the three programs, which we cannot completely explain,
because of shared-memory effects and highly optimized implementations
of PFFT and FFTW incorporating multiple algorithms for global communication.
FFTU outperforms heFFTe with the Intel MKL library as sequential kernel, 
but this might be different for heFFTe with FFTW as its sequential kernel. 

The most important application of our algorithm will most likely be the 3D case,
because of our single all-to-all communication superstep, where PFFT and FFTW need two or even three
in case the output distribution needs to be the same as the input distribution.
Another important application would be the case where the input array is very rectangular,
with one dimension much larger in size than the others. Here, the advantage is better scalability,
because we can still use $\sqrt{N}$ processors, where the other methods are limited by the size of the smallest dimensions.

\section{Future work}

Our package FFTU performs multidimensional complex-to-complex fast Fourier transforms in the cyclic distribution.
For future work, this could be extended to related transforms such as the real-to-complex fast Fourier transform (RFFT),
the discrete sine transform (DST), and the discrete cosine transform (DCT). We can do this by combining existing 1D algorithms dimensionwise~\cite{koopman22}.
To exploit the symmetry of these transforms, the \emph{zig-zag cyclic distribution}~\cite{bisseling20,inda01b}
could be used instead of the cyclic distribution.

All our parallelism is explicit and based on the distributed-memory approach using MPI,
and we do not yet exploit the possible improvements that shared-memory parallelism has to offer. With 128 cores on a processor node, this may yield significant
improvements in our communication superstep. Using OpenMP, or using some of the communication schemes of FFTW for this purpose
could be beneficial. Another research direction would be to exploit better methods for performing the single
communication superstep in our algorithm, either by letting the transposition algorithms of FFTW handle this superstep,
or by using improved redistribution methods~\cite{dalcin19} instead of our vanilla call to MPI.

Some applications perform elementwise multiplications in both domains of the Fourier transform.
In the solution of the time-dependent Schr\"{o}dinger equation on a multidimensional grid,
a wave function is propagated in time by multiplying it pointwise by a potential function in the time domain,
and by a momentum operator in the frequency domain.
This means that we only need  one all-to-all communication superstep
per (forward or inverse) transform. No further communication is required in the propagation.

In other applications, however, such as classical molecular dynamics, the data may have to be stored in a block distribution instead of a cyclic distribution,
to accommodate parts of the application program outside the FFT. This may require additional data redistribution;
further investigation is needed to see how this can be avoided or mitigated.

\section*{Acknowledgments}
We thank SURFsara in Amsterdam, the Netherlands, for giving us access to the supercomputer Snellius
and for help in running our experiments. We also thank the Dutch Research Council (NWO) for funding the computing time
we used  at SURFsara, under project EINF-1158.

\bibliographystyle{abbrvurl}
\bibliography{refs.bib}

\begin{thebibliography}{10}

\bibitem{ayala20}
A.~Ayala, S.~Tomov, A.~Haidar, and J.~Dongarra.
\newblock {{heFFTe}: Highly Efficient {FFT} for Exascale}.
\newblock In V.~V. Krzhizhanovskaya, G.~Z{\'a}vodszky, M.~H. Lees, J.~J.
  Dongarra, P.~M.~A. Sloot, S.~Brissos, and J.~Teixeira, editors, {\em
  Computational Science -- ICCS 2020}, volume 12137 of {\em Lecture Notes in
  Computer Science}, pages 262--275. Springer, Cham, 2020.
\newblock \href {https://doi.org/10.1007/978-3-030-50371-0_19}
  {\path{doi:10.1007/978-3-030-50371-0_19}}.

\bibitem{bisseling20}
R.~H. Bisseling.
\newblock {\em Parallel Scientific Computation: A Structured Approach Using
  {BSP}}.
\newblock Oxford University Press, Oxford, UK, second edition, 2020.
\newblock \href {https://doi.org/10.1093/oso/9780198788348.001.0001}
  {\path{doi:10.1093/oso/9780198788348.001.0001}}.

\bibitem{blackford97}
L.~S. Blackford, J.~Choi, A.~Cleary, E.~D'Azevedo, J.~Demmel, I.~Dhillon,
  J.~Dongarra, S.~Hammarling, G.~Henry, A.~Petitet, K.~Stanley, D.~Walker, and
  R.~C. Whaley.
\newblock {\em {ScaLAPACK} Users' Guide}.
\newblock SIAM, Philadelphia, 1997.
\newblock \href {https://doi.org/10.1137/1.9780898719642}
  {\path{doi:10.1137/1.9780898719642}}.

\bibitem{borowski04}
S.~Borowski and T.~{Kl\"{u}ner}.
\newblock Massively parallel {Hamiltonian} action in pseudospectral algorithms
  applied to quantum dynamics of laser induced desorption.
\newblock {\em Chem. Phys.}, 304(1):51--58, 2004.
\newblock \href {https://doi.org/10.1016/j.chemphys.2004.06.012}
  {\path{doi:10.1016/j.chemphys.2004.06.012}}.

\bibitem{cooley65}
J.~W. Cooley and J.~W. Tukey.
\newblock An algorithm for the machine calculation of complex {F}ourier series.
\newblock {\em Math. Comp.}, 19:297--301, 1965.
\newblock \href {https://doi.org/10.1090/S0025-5718-1965-0178586-1}
  {\path{doi:10.1090/S0025-5718-1965-0178586-1}}.

\bibitem{dalcin19}
L.~Dalcin, M.~Mortensen, and D.~E. Keyes.
\newblock Fast parallel multidimensional {FFT} using advanced {MPI}.
\newblock {\em J. Parallel Distrib. Comput.}, 128:137--150, 2019.
\newblock \href {https://doi.org/10.1016/j.jpdc.2019.02.006}
  {\path{doi:10.1016/j.jpdc.2019.02.006}}.

\bibitem{ding95}
C.~H.~Q. Ding, R.~D. Ferraro, and D.~B. Gennery.
\newblock A portable {3D FFT} package for distributed-memory parallel
  architectures.
\newblock In {\em Proceedings of the Seventh {SIAM} Conference on Parallel
  Processing for Scientific Computing}, pages 70--71. {SIAM}, 1995.

\bibitem{foster97}
I.~T. Foster and P.~H. Worley.
\newblock Parallel algorithms for the spectral transform method.
\newblock {\em SIAM J. Sci. Comput.}, 18(3):806–837, 1997.
\newblock \href {https://doi.org/10.1137/S1064827594266891}
  {\path{doi:10.1137/S1064827594266891}}.

\bibitem{frigo05}
M.~Frigo and S.~G. Johnson.
\newblock The design and implementation of {FFTW3}.
\newblock {\em Proc. {IEEE}}, 93(2):216--231, 2005.
\newblock \href {https://doi.org/10.1109/JPROC.2004.840301}
  {\path{doi:10.1109/JPROC.2004.840301}}.

\bibitem{inda01}
M.~A. Inda and R.~H. Bisseling.
\newblock A simple and efficient parallel {FFT} algorithm using the {BSP}
  model.
\newblock {\em Parallel Comput.}, 27(14):1847--1878, 2001.
\newblock \href {https://doi.org/10.1016/S0167-8191(01)00118-1}
  {\path{doi:10.1016/S0167-8191(01)00118-1}}.

\bibitem{inda01b}
M.~A. Inda, R.~H. Bisseling, and D.~K. Maslen.
\newblock On the efficient parallel computation of {L}egendre transforms.
\newblock {\em SIAM J. Sci. Comput.}, 23(1):271--303, 2001.
\newblock \href {https://doi.org/10.1137/S1064827599355864}
  {\path{doi:10.1137/S1064827599355864}}.

\bibitem{jung16}
J.~Jung, C.~Kobayashi, T.~Imamura, and Y.~Sugita.
\newblock Parallel implementation of {3D FFT} with volumetric decomposition
  schemes for efficient molecular dynamics simulations.
\newblock {\em Comput. Phys. Commun.}, 200:57--65, 2016.
\newblock \href {https://doi.org/10.1016/j.cpc.2015.10.024}
  {\path{doi:10.1016/j.cpc.2015.10.024}}.

\bibitem{koopman22}
T.~Koopman.
\newblock {The Tensor product of Bulk Synchronous Parallel Algorithms}.
\newblock Master's thesis, Utrecht University, Utrecht, The Netherlands, Jan.
  2022.

\bibitem{kosloff88}
R.~Kosloff.
\newblock Time-dependent quantum-mechanical methods for molecular dynamics.
\newblock {\em J. Phys. Chem.}, 92:2087--2100, 1988.
\newblock \href {https://doi.org/10.1021/j100319a003}
  {\path{doi:10.1021/j100319a003}}.

\bibitem{leforestier91}
C.~Leforestier, R.~H. Bisseling, C.~Cerjan, M.~D. Feit, R.~Friesner,
  A.~Guldberg, A.~Hammerich, G.~Jolicard, W.~Karrlein, H.-D. Meyer, N.~Lipkin,
  O.~Roncero, and R.~Kosloff.
\newblock A comparison of different propagation schemes for the time dependent
  {S}chr{\"{o}}dinger equation.
\newblock {\em J. Comput. Phys.}, 94(1):59--80, 1991.
\newblock \href {https://doi.org/10.1016/0021-9991(91)90137-A}
  {\path{doi:10.1016/0021-9991(91)90137-A}}.

\bibitem{li10}
N.~Li and S.~Laizet.
\newblock A highly scalable {2D} decomposition library and {FFT} interface.
\newblock In {\em Proceedings of the Cray User Group 2010 Conference}, pages
  1--13, 2010.

\bibitem{pekurovsky12}
D.~Pekurovsky.
\newblock {P3DFFT}: A framework for parallel computations of {Fourier}
  transforms in three dimensions.
\newblock {\em SIAM J. Sci. Comput.}, 34(4):C192--C209, 2012.
\newblock \href {https://doi.org/10.1137/11082748X}
  {\path{doi:10.1137/11082748X}}.

\bibitem{pippig13}
M.~Pippig.
\newblock {PFFT}: An extension of {FFTW} to massively parallel architectures.
\newblock {\em SIAM J. Sci. Comput.}, 35(3):C213 -- C236, 2013.
\newblock \href {https://doi.org/10.1137/120885887}
  {\path{doi:10.1137/120885887}}.

\bibitem{plimpton18}
S.~Plimpton, A.~Kohlmeyer, P.~Coffman, and P.~Blood.
\newblock {fftMPI}, a library for performing 2d and 3d {FFTs} in parallel,
  2018.
\newblock \href {https://doi.org/10.11578/dc.20201001.68}
  {\path{doi:10.11578/dc.20201001.68}}.

\bibitem{plimpton97}
S.~Plimpton, R.~Pollock, and M.~Stevens.
\newblock Particle-mesh {Ewald} and {rRESPA} for parallel molecular dynamics
  simulations.
\newblock In {\em Proceedings of the Eighth {SIAM} Conference on Parallel
  Processing for Scientific Computing}. {SIAM}, Philadelphia, 1997.

\bibitem{popovici20}
D.~T. Popovici, M.~D. Schatz, F.~Franchetti, and T.~M. Low.
\newblock A flexible framework for multidimensional {DFTs}.
\newblock {\em SIAM J. Sci. Comput.}, 42(5):C245--C264, 2020.
\newblock \href {https://doi.org/10.1137/19M1288401}
  {\path{doi:10.1137/19M1288401}}.

\bibitem{pryor17}
A.~Pryor~Jr., Y.~Yang, A.~Rana, M.~Gallagher-Jones, J.~Zhou, Y.~H. Lo,
  G.~Melinte, W.~Chiu, J.~A. Rodriguez, and J.~Miao.
\newblock {GENFIRE}: A generalized {F}ourier iterative reconstruction algorithm
  for high-resolution {3D} imaging.
\newblock {\em Sci. Rep.}, 7, 2017.
\newblock \href {https://doi.org/10.1038/s41598-017-09847-1}
  {\path{doi:10.1038/s41598-017-09847-1}}.

\bibitem{puschel05}
M.~{Puschel}, J.~M.~F. {Moura}, J.~R. {Johnson}, D.~{Padua}, M.~M. {Veloso},
  B.~W. {Singer}, {Jianxin Xiong}, F.~{Franchetti}, A.~{Gacic}, Y.~{Voronenko},
  K.~{Chen}, R.~W. {Johnson}, and N.~{Rizzolo}.
\newblock {SPIRAL}: Code generation for {DSP} transforms.
\newblock {\em Proc. {IEEE}}, 93(2):232--275, 2005.
\newblock \href {https://doi.org/10.1109/JPROC.2004.840306}
  {\path{doi:10.1109/JPROC.2004.840306}}.

\bibitem{valiant90}
L.~G. Valiant.
\newblock A bridging model for parallel computation.
\newblock {\em Comm. {ACM}}, 33(8):103--111, 1990.
\newblock \href {https://doi.org/10.1145/79173.79181}
  {\path{doi:10.1145/79173.79181}}.

\bibitem{vanloan92}
C.~Van~Loan.
\newblock {\em Computational Frameworks for the Fast Fourier Transform}.
\newblock SIAM, Philadelphia, 1992.
\newblock \href {https://doi.org/10.1137/1.9781611970999}
  {\path{doi:10.1137/1.9781611970999}}.

\bibitem{yzelman14}
A.~N. Yzelman, R.~H. Bisseling, D.~Roose, and K.~Meerbergen.
\newblock Multicore{BSP} for {C}: a high-performance library for shared-memory
  parallel programming.
\newblock {\em Internat. J. Parallel Programming}, 42(4):619--642, 2014.
\newblock \href {https://doi.org/10.1007/s10766-013-0262-9}
  {\path{doi:10.1007/s10766-013-0262-9}}.

\end{thebibliography}

\end{document}